\documentclass[%
 reprint,
 superscriptaddress,
 amsmath,amssymb,
 aps,
pra,
floatfix,
]{revtex4-1}

\usepackage{graphicx}
\usepackage{dcolumn}
\usepackage{bm}
\usepackage{hyperref}

\usepackage{cleveref}
\usepackage{tikz}
\usetikzlibrary{decorations.markings}
\usepackage[caption=false]{subfig}
\usepackage{amsthm,dsfont,adjustbox}
\usepackage[normalem]{ulem}
\usepackage{pdfpages}

\makeatletter
\AtBeginDocument{\let\LS@rot\@undefined}
\makeatother

\newtheorem{prop}{Proposition}[section]

\newtheorem{lem}{Lemma}[section]

\newcommand{\Loop}{\mathrm{Loop}}
\newcommand{\PT}{\mathcal{PT}}
\renewcommand{\P}{\mathcal{P}}
\newcommand{\T}{\mathcal{T}}
\newcommand{\Z}{\mathds{Z}}
\newcommand{\R}{\mathds{R}}

\newcommand\inv[1]{#1\raisebox{1.15ex}{$\scriptscriptstyle-\!1$}}

\begin{document}

\title{Non-abelian nature of systems with multiple exceptional points}

\author{Eric J. Pap}
 \affiliation{Van Swinderen Institute for Particle Physics and Gravity, University of Groningen, Nijenborgh 4, NL-9747 AG Groningen, The Netherlands}
 \affiliation{Johann Bernoulli Institute, University of Groningen, Nijenborgh 9, NL-9747 AG Groningen, The Netherlands}
\author{Dani\"el Boer}
 \affiliation{Van Swinderen Institute for Particle Physics and Gravity, University of Groningen, Nijenborgh 4, NL-9747 AG Groningen, The Netherlands}%
\author{Holger Waalkens}
 \affiliation{Johann Bernoulli Institute, University of Groningen, Nijenborgh 9, NL-9747 AG Groningen, The Netherlands}%

\date{\today}

\begin{abstract}
The defining characteristic of an exceptional point (EP) in the parameter space of a family of operators is that upon encircling the EP eigenstates are permuted.
In case one encircles multiple EPs, the question arises how to properly compose the effects of the individual EPs. This was thought to be ambiguous. We show that one can solve this problem by considering based loops and their deformations. The theory of fundamental groups allows to generalize this technique to arbitrary degeneracy structures like exceptional lines in a three-dimensional parameter space.
As permutations of three or more objects form a non-abelian group, the next question that arises is whether one can experimentally demonstrate this non-commutative behavior. 
This requires at least two EPs of a family of operators that have at least 3 eigenstates.
A concrete implementation in a recently proposed $\PT$ symmetric waveguide system is suggested as an example of how to experimentally check the composition law and show the non-abelian nature of non-hermitian systems with multiple EPs.
\end{abstract}

\maketitle

\section{Introduction}
Exceptional points (EPs) are typically considered in systems with a discrete set of eigenstates. The exchange of eigenstates when traversing a closed loop around an EP is its defining characteristic (see e.g.\ \cite{heiss2012physics}). 
The term 'exceptional' was originally used to indicate the presence of a degeneracy in the sense that two or more eigenvalues or levels coincide at an EP (cf.\ \cite{kato}). 
At an EP the characteristic polynomial of the operator has a higher order zero.
Such a degeneracy could arise from a branch point, allowing for permutations of eigenvalues upon encircling.
We take the latter property to define an EP as a degeneracy of a (matrix) operator family such that non-trivial permutations of eigenvalues occur upon following the eigenvalues along a closed loop around the degeneracy. The non-trivial branch structure implies that the matrix family is non-hermitian where at the EP the operator cannot be diagonalized.

Often the concept of $\PT$ symmetry \cite{bender1998real,bender2005introduction} is treated together with EPs. This has two main origins; in physics $\PT$ symmetry is often considered a replacement of hermiticity, and EPs usually mark points where the $\PT$ symmetry becomes spontaneously broken ($\PT$ phase transitions).  Also, one may check that $\PT$ phase transitions share the higher order zero condition with EPs.
However, it turned out that a $\PT$ symmetric system is sometimes equivalent to a hermitian system  \cite{mostafa3,mostafazadeh2003exact}.
One can say that interesting aspects of $\PT$ symmetry may arise at an EP but $\PT$ symmetry is not the main framework to study EPs.

Studies of EPs started primarily with EP2s (e.g.\ \cite{heiss2012physics}), i.e.~EPs where two eigenvalues  are interchanged. Their characteristics are now well-understood. Along a closed loop around an EP2 in the parameter plane two eigenvectors are exchanged with one acquiring a minus sign. This has also been verified experimentally \cite{dembowski2001experimental}. Hence one has to encircle an EP2 at least four times to identically map the eigenvectors, whereas the map of the eigenvalues only needs two turns to become the identity because the sign of the eigenvectors is then irrelevant. This resembles the characteristics of a diabolical point (DP), which is a degeneracy where the eigenvectors remain linearly independent. Upon encircling a DP the eigenvectors are mapped to minus themselves, and an EP2 can therefore be considered to be 'half' of a DP \cite{keck2003unfolding}. 

Recently, EPs with 3 coalescing levels (EP3s) have become of interest. They were already studied explicitly in \cite{demange2011signatures}, and now actual experiments are investigated (cf.\ \cite{schnabel2017ptep3} and refs therein). Here waveguides with gain and loss regions  
provide an optical system that is formally equivalent to a quantum system with a non-hermitian Hamiltonian. Such systems have already been introduced earlier to experimentally study aspects of $\PT$ symmetry \cite{ruter2010observation}. 

In this paper we focus on the interplay of multiple EPs which naturally leads to the question of how the permutations obtained from encircling two or more EPs is related to the permutations associated with loops around the individual EPs.  
This problem was already analyzed in \cite{analysismultEP2012kim} for systems with three levels. It was concluded that a composition of permutations associated with individual EPs cannot be done as the order of the permutations was ambiguous. In this paper we show that a definite answer can be obtained using {\it based loops}, which are oriented loops starting and ending at a fixed base point, and the continuous deformations of these loops as they enter the definition of fundamental groups. 
This will allow us in particular to study non-abelian effects which arise for systems with three or more levels that have two or more exceptional points. 
We illustrate an experimental implementation to test the results using a waveguide system.

The paper is divided into a theoretical analysis of the problem in \cref{sec:theory}, and explicit application in \cref{sec:waveguide_experiment}. In \cref{sec:theory}, we discuss the mathematical background required to deal with multiple EPs, how to solve the composition problem and address the resulting non-abelian effects. Then in \cref{sec:waveguide_experiment} we discuss a wave-guide experiment which allows one to observe this non-abelian nature of multiple EPs.
We end with a summary in \cref{sec:summary}.


\section{Theoretical analysis}
\label{sec:theory}


\subsection{The problem}\label{sec:problem}
Let us consider a finite-dimensional quantum problem given by a parametrized family of $n\times n$ matrices. 
We do not impose any condition on the matrices, e.g.\ they may be non-hermitian. 
If the eigenvalues locally follow an $N$-sheet branch structure ($N\leq n$), the corresponding branch point is called an EP$N$. The branch structure is then similar to that of the $N$th complex root, revealing cyclic permutations of order $N$. Because of the cyclic property, non-abelian behavior can never occur using a single EP, whatever its order. We note that the eigenvectors may acquire a phase (cf.\ e.g.\  \cite{heiss2012physics}). However, we will disregard phases and consider permutations of eigenvalues only.

We can now state our main question in a more precise fashion.
Consider two EPs encircled by two oriented loops $\gamma_1$ and $\gamma_2$, respectively, as depicted in \cref{fig:theoretical_setup}.
 Suppose one has measured the permutations obtained from the loops $\gamma_1$ and $\gamma_2$. 
Which permutation should one obtain for a loop $\gamma_3$ encircling both EPs?

\begin{figure}[h]
    \centering
    \includegraphics{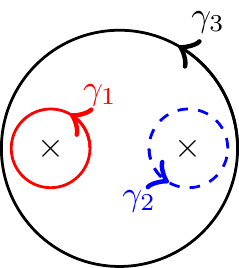}
    \caption{ 
    Two EPs (crosses) encircled individually by  loops $\gamma_1$ and $\gamma_2$, respectively,  and a loop $\gamma_3$ which encircles both EPs.}
    \label{fig:theoretical_setup}
\end{figure}


\subsection{Basepoint solution}
The essential insight is that one first needs to fix a common base point for the paths $\gamma_1$, $\gamma_2$ and $\gamma_3$; the permutations should be compared for the same initial system parameters. Let us denote by $\Delta$ the discriminant set of the family of operators, i.e.~$\Delta$ is the set of parameters for which two or more eigenvalues coincide, and by $X$ the complement of $\Delta$. Fixing a basepoint $x_0\in X$ we can consider the measurement paths that start and end at $x_0$, that is, the loops based at $x_0$. Let us denote by $\Loop(x_0)$ the set of oriented loops in $X$ that are based at $x_0$. 
As each loop in $\Loop(x_0)$ is contained in $X$ the operators have distinct eigenvalues at any point on such a loop. Tracing a loop $\gamma\in\Loop(x_0)$ induces a fixed change of eigenstates, in particular it induces a permutation $p_\gamma$ of the eigenvalues.
Denote by 
\begin{equation} \label{eq:def_lambda_group}
    \Lambda(x_0)=\{p_\gamma \;| \;\gamma\in\Loop(x_0)\}
\end{equation}
the group of permutations that can be achieved in such a way. This is a subgroup of the symmetric group
of the $n$ distinct eigenvalues, and by using a labeling is isomorphic to a subgroup of $S_n$. 

A group like $\Lambda(x_0)$ was already mentioned in the book by Kato \cite{kato} where the term exceptional point was used for the first time, and the group was called the $\lambda$-group. The $\lambda$-group there consists of the permutations that arise from analytically continuing the eigenvalues back to some initial point. The group $\Lambda(x_0)$ is a generalization by allowing for more general adiabatic connections; details of such a geometric connection can be found in \cite{moshdgeo}. To describe $\Lambda(x_0)$ more rigorously; by parallel transport each loop $\gamma\in\Loop(x_0)$ induces a linear map on state space, which by the adiabatic assumption maps eigenstates to eigenstates. The group $\Lambda(x_0)$ is then obtained from the holonomy group at $x_0$ by restricting to the permutations of the eigenstates.
We point out that the formalism does not require the complex analytic theory of Riemann sheets. It is sufficient that the eigenvalues vary smoothly with system parameters. This allows us to use operators that also involve complex conjugates which are excluded in the complex analytic case.

The concatenation of two loops in $\Loop(x_0)$ defines a 'product' in $\Loop(x_0)$ that is in general non-abelian. By the holonomy interpretation, the assignment $\gamma\mapsto p_\gamma$ preserves this product in the sense that
\begin{equation}
    p_{\gamma_2\gamma_1}=p_{\gamma_2}\circ p_{\gamma_1}
\end{equation}
where in $\gamma_2\gamma_1$ we first track $\gamma_1$ and then $\gamma_2$.

It is at this point that the distinction between $\gamma_2\gamma_1$ and $\gamma_1\gamma_2$ becomes interesting. This is because we are mapping the loops to a permutation group that in general is non-commutative.
In fact, as (based) loops the two products $\gamma_2\gamma_1$ and $\gamma_1\gamma_2$  may be different in the sense that they are not homotopic relative to the basepoint. 
The loops $\gamma_1 $ and $ \gamma_2$ both start and end at the basepoint $x_0$. The concatenation $\gamma_2\gamma_1$ is a loop that starts at $x_0$ and following $\gamma_1$ intermediately comes back to $x_0$, after which $\gamma_2$ is traversed which again ends at $x_0$.  A homotopic deformation of $\gamma_2\gamma_1$ as a loop in $\Loop(x_0)$ is a continuous deformation of the concatenation $\gamma_2\gamma_1$ within $X$ that keeps the starting point of $\gamma_1$ and the end point of $\gamma_2$ fixed; the intermediate visit of $x_0$ becomes irrelevant. This applies analogously to the product $\gamma_1\gamma_2$.
In \cref{fig:basepoint_on_loop} we show continuous deformations of $\gamma_1\gamma_2$ and $\gamma_2\gamma_1$. We in particular see that one cannot deform $\gamma_2\gamma_1$ to $\gamma_1\gamma_2$ within $X$ if one needs to keep the basepoint fixed.

\begin{figure}
    \subfloat[Two paths around two EPs.
    \label{sfig:twopaths}
    ]{
    \includegraphics{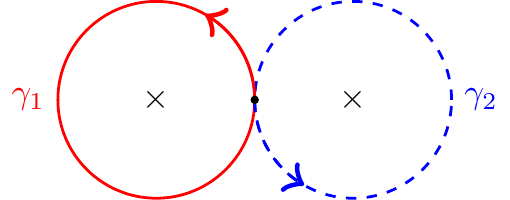}
    }
        
    \subfloat[Deformed \(\gamma_1\gamma_2\).
    \label{sfig:g1g2}]{
    \includegraphics{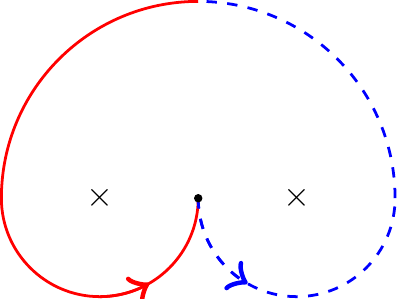}
    }
    
    \subfloat[Deformed \(\gamma_2\gamma_1\).
    \label{sfig:g2g1}]{
    \includegraphics{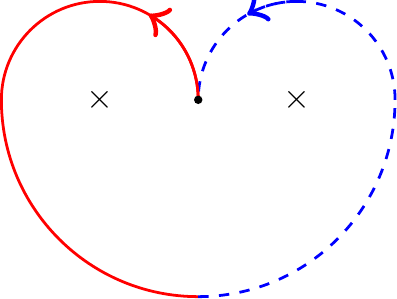}
    }
    
    \caption{Example of loops based at the bold dot that enclose two EPs marked as crosses. The deformed $\gamma_1\gamma_2$ and $\gamma_2\gamma_1$ resemble each other, but are not homotopic relative to the base point.}
    \label{fig:basepoint_on_loop}
\end{figure}

The continuous deformation of the based loops $\gamma\in \Loop(x_0)$ and their concatenation leads to the definition of the fundamental group $\pi_1(X,x_0)$ whose elements are the equivalence classes $[\gamma]$ of loops that are homotopic to a representative $\gamma$ and where the group operation is defined by the product given by the concatenation of loops.
The situation depicted in \cref{fig:basepoint_on_loop} is then general. By the theory of fundamental groups, once fundamental paths are chosen, any loop can be written in terms of these.
So far we used deformations to stress the non-commutativity of the product of two based loops. In the next subsection we will see that deformations are also relevant for the concrete question of calculating permutations.

Let us now come back to the question posed in \cref{sec:problem}, see \cref{fig:theoretical_setup}. 
First, we choose basepoints $x_1$, $x_2$ for the small loops $\gamma_1$ and $\gamma_2$, respectively,  and $x_3$ for the big loop $\gamma_3$. The basepoints $x_1$ and $x_2$ are likely to be different. In this case we choose an oriented path $b$ from $x_1$ to $x_2$ as shown in \cref{fig:connecting_path} which allows us relate the based loops $\gamma_1$ and $\gamma_2$ in the sense 
that a loop $\gamma_2\in\Loop(x_2)$ can be associated with a loop $\inv{b}\gamma_2 b\in\Loop(x_1)$. For convenience, we will refer to this operation as pull-back via $b$, and call $b$ a bridge from $x_1$ to $x_2$. 

As shown in \cref{fig:relabelling_path}, a bridge induces a fixed labelling of eigenvalues, and permutations must be rewritten accordingly. This means, if one wants to talk about the permutation 'first $\gamma_1$, then $\gamma_2$' one has to pick basepoints, and if these do not coincide also bridge(s). 
For the relabelling illustrated in \cref{fig:relabelling_path}, suppose, e.g., that $\gamma_2$ induces the permutation $(1'2'3')$ (where we use the cycle notation). Then pulling-back via $b$ one obtains $1\mapsto 2'\mapsto 3'\mapsto3$, $2\mapsto1'\mapsto2'\mapsto1$ and $3\mapsto3'\mapsto 1'\mapsto2$, that is, $\inv{b}\gamma_2b$ induces $(132)$.

\begin{figure}
    \subfloat[Two loops with different basepoints. A bridge is needed to pull-back information from one point to the other. This requires that the bridge is traversed twice and in opposite directions.\label{fig:connecting_path}]{
    \includegraphics{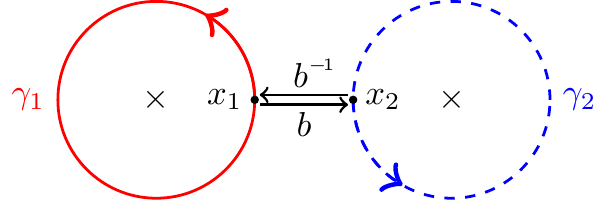}
    }
    
    \subfloat[Example of eigenvalue paths induced from a bridge.
    Open circles mark the eigenvalues at $x_1$ and have unprimed labels, filled circles mark eigenvalues at $x_2$ and have primed labels. The bridge induces the relabelling $1\mapsto 2'$, $2\mapsto 1'$, $3\mapsto3'$.\label{fig:relabelling_path}
    ]{\includegraphics{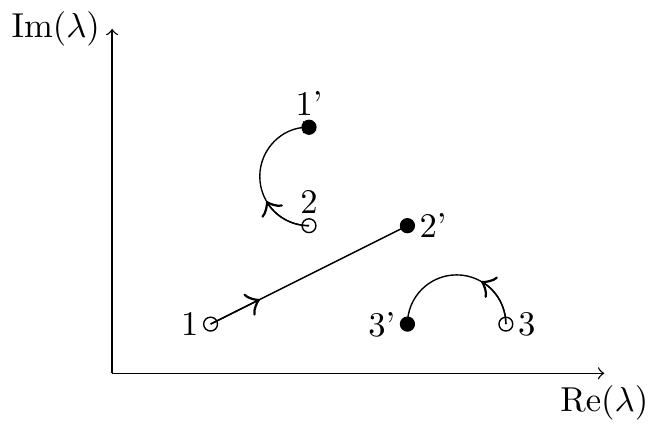}
    }
    
    \caption{Aspects of a connecting path or bridge, both in parameter and eigenvalue space.}
    \label{fig:aspects_connecting_path}
\end{figure}

Using a bridge from $x_1$ to $x_3$, the pull-back allows us to continue as if $x_3=x_1$, a choice that would be convenient in experiment as well. By the previous discussion, the big loop can be decomposed into the smaller loops, and composing permutations accordingly yields a unique permutation for the big loop, given the labeling and bridges used. Observe that this technique may also be used to keep track of the occurring geometric phases. The various dependencies that occur we discuss in \cref{sec:remarks}.


\subsection{Permutations are topological}
We now turn to the fact that permutations induced by loops around EPs are topological in nature, as opposed to geometric. This means that based loops that are homotopic induce the same permutation of eigenvalues. This does not assume anything on the nature of the degeneracies, i.e.\ whether they are EPs, DPs or yet another type. Note that this fact also allows one to pick the most convenient loop in a homotopy class, without any theoretical requirements on the quantum system.

\begin{lem}
    If $\gamma,\tilde\gamma\in\Loop(x_0)$ are homotopic relative to $x_0$, then the induced permutations are equal, i.e.\ $p_\gamma=p_{\tilde\gamma}$. In other words, the assignment $\gamma\mapsto p_\gamma$ factors as
    \begin{equation}
        \begin{split}
            \Loop(x_0)\to \; \pi_1(&X,x_0) \to \Lambda(x_0)\\
            \gamma\mapsto &[\gamma]\mapsto p_\gamma
        \end{split}
    \end{equation}
    where each map preserves products.
\end{lem}
\begin{proof}
We already remarked that a labeling induces an isomorphism between $\Lambda(x_0)$ and a subgroup of $S_n$. As $S_n$ is discrete, so is $\Lambda(x_0)$. The homotopy from $\gamma$ to $\tilde\gamma$ induces a homotopy from $p_\gamma$ to $p_{\tilde\gamma}$, which by discreteness is constant.
\end{proof}

Let us propose a procedure for checking the composition rule in the situation of a planar parameter space, where we consider a loop that encircles $k$ EPs, each with winding number 1 which intuitively means that each EP is encircled exactly once. Homotopy theory allows to extend such a procedure to higher dimensional parameter spaces where the degeneracies are of codimension 2 like the exceptional lines in the three-level system that we consider in \cref{sec:waveguide_experiment}. A measurement could proceed according to the following steps
\begin{enumerate}
    \item fix a loop $\gamma$ encircling all EPs once, and choose a base point $x_0$ on this loop,
    \item write $[\gamma]=[\gamma_k]\cdots [\gamma_1]$ where each $\gamma_i\in \Loop(x_0)$ encircles a single EP with winding number 1,
    \item measure the permutations $p_i:=p_{\gamma_i}$ and $p:=p_\gamma$,
    \item check $p$ and $p_k\cdots p_1$ for equality.
\end{enumerate}

Non-abelian behavior occurs if there are two loops $\gamma_1,\gamma_2$ such that $p_{\gamma_2}\circ p_{\gamma_1}\ne p_{\gamma_1}\circ p_{\gamma_2}$. As $S_n$ is commutative for $n<3$, a system in which this is possible should have $n\geq 3$ many levels. Note that EP3s are not required to see non-abelian behaviour.  Instead it is sufficient to have a system with three levels and two EP2s with permutations $(12)$ and $(23)$, respectively.

Another observable property is orientation dependence, e.g.\ by comparing a loop encircling two EPs with a figure 8 shaped partner loop, or more precisely, compare the permutation along the loop $\gamma_2\gamma_1$ shown in \cref{fig:basepoint_on_loop} with, e.g., $\inv{\gamma_2}\gamma_1$. As opposed to the previous construction the present one requires EP$N$s with $N\ge3$. This is due to the fact that for an EP2 the permutation is always a transposition and hence equals its own inverse.


\subsection{Examples}

Let us determine the $\Lambda$-group for some well-known cases. Encircling a single EP$N$ once yields an $N$-cycle, and one has $\Lambda(x_0)\cong \Z/N\Z$. This identification does not depend on $x_0$ if the parameter space is path-connected, which we assume for simplicity, but the precise eigenvalues permuted do depend on $x_0$. 
Consider now 2 EP2s, then there are a number of different possibilities depending on how the eigenvalue sheets are connected and whether the EPs are not located at the same point in parameter space, that is, if they can be circumscribed individually.

Suppose the EP2s share no sheet. Then the system must have at least 4 distinct eigenvalues, and we may take a 4-dimensional (sub)system. Fixing a basepoint and a labelling, we may assume that the EPs have permutations $(12)$ resp.\ $(34)$. If the EPs are at different locations, we can permute independently and one has $\Lambda(x_0)\cong S_2\times S_2$. In case the EPs are on top of each other, the only non-trivial permutation is $(12)(34)$, so $\Lambda(x_0)\cong S_2$.

Suppose the EP2s share one sheet. Then the system can be taken 3-dimensional, and the permutations as $(12)$ and $(23)$. These two transpositions generate $S_3$, and hence $\Lambda(x_0)\cong S_3$ as this is the largest group that can be obtained with a 3 dimensional system. We note that encircling both EP2s in the right order yields $(12)(23)=(123)$, and as \cite{analysismultEP2012kim} showed, doing this 3 times yields the identity. For another loop, one may have the opposite order and measure $(23)(12)=(132)$, again a 3-cycle. However, we stress that using the theory introduced in this paper we can calculate the outcome after encircling just once. This is crucial in showing non-abelian behavior as this manifests itself in the difference between $(123)$ and $(132)$, which are both 3-cycles. To conclude this case, if the EP2s would lie on top of each other the resulting structure would look like an EP3, which we treated above.

Finally, suppose the EP2s share both sheets, as happens in the standard case (e.g.\ \cite{heiss2012physics}). Now the system can be taken 2-dimensional with both permutations equal to $(12)$. Hence $\Lambda(x_0)\cong S_2$, similar to taking just 1 EP, and note that the EPs cannot be on top of each other without becoming equal.

We see that the $\Lambda$-group detects the differences in sheet structure. If we include more EPs or allow higher order EP$N$s one can reason similarly, be it with more involved permutations.


\subsection{Remarks}
\label{sec:remarks}

Let us inspect how the exposition above depends on choices such as basepoints and bridges. We remark that this is similar to the discussion that two fundamental groups
$\pi_1(X,x_0)$ and $\pi_1(X,x_1)$ with different base points $x_0$ and $x_1$ are isomorphic by a conjugation-like construction, the conjugacy provided by a bridge between $x_0$ and $x_1$.

Concerning the basepoint, choosing one fixes a path-connected component of $X$. Within this component, bridges can be used to connect different basepoints, relating $\Loop$-spaces by conjugation (for the eigenvalues, standard bookkeeping of the labels appears). 
We disregard the case where $X$ is not path-connected as it is in general not meaningful to compare levels associated with parameters in different connected components of $X$ because of the absence of a continuous dependence of the levels on the parameters.

Given two bridges $b,\tilde{b}$ between the same basepoints, the results of pull-back of a loop may very well differ. Key is the loop $\inv{\tilde{b}}b\in\Loop(x_0)$, which may yield a non-trivial permutation. Indeed, $b$ and $\tilde{b}$ may pass an EP on different sides, such that we approach the final point using different sheets, where the loop will indeed reveal this EP permutation. Again, there is uniqueness up to conjugation, as made precise in the next lemma.

\begin{lem}
    Let $x_0$ and $x_1$ be basepoints, let $b,\tilde{b}$ be bridges from $x_0$ to $x_1$. The two pull-back operations are related by conjugation with the permutation of $\inv{b}\tilde{b}$.
\end{lem}
\begin{proof}
    For \(\gamma\in\Loop(x_1)\) arbitrary, one has \(\inv{\tilde{b}}\gamma\tilde{b}\) homotopic to \(\inv{(\inv{b}\tilde{b})}(\inv{b}\gamma b)(\inv{b}\tilde{b})\), where all factors are in $\Loop(x_0)$. The claim now follows.
\end{proof}

Note that conjugation in $S_n$ leaves the cycle structure invariant, so one may think that a permutation depends only on the loop (and by the above, only the homotopy class). This holds true for the cycle structure, but one should still be careful when using concrete labeling, which varies even per basepoint.

Let us also discuss coordinate dependence. When reparametrizing parameter space, we assume that the reparametrization establishes a homeomorphism of the original parameter space. 
This induces a homeomorphism of the non-degeneracy space $X$, and so loops in one parametrization correspond to loops in the other. Also here a conjugation-like correspondence appears. It does supply another reason that deformations, even of the degeneracies themselves, do not change the physical aspects.

We emphasize that the exposition does not include any assumptions on the operators. However, it is well-known that hermitian systems do not allow for EPs. This is usually proven by the non-existence of a complete set of orthonormal eigenstates \emph{at} an EP. Using the techniques above, we may provide a more topological proof, where we only need to look \emph{around} the EP. More concretely, one may show non-existence of EPs by showing that $\Lambda(x_0)$ is trivial, as done in the next proposition. The premise is satisfied for any hermitian family and also includes exact $\PT$-symmetric systems \cite{mostafazadeh2003exact}.

\begin{prop}
    Let $T(x)$ be a family of $n\times n$ matrix operators. If $T(x)$ has real eigenvalues for any $x\in X$, then $\Lambda(x_0)=0$ for all $x_0\in X$.
\end{prop}
\begin{proof}
    Let $\gamma\in\Loop(x_0)$ be any loop in $X$, denote by $\lambda_i(t)$ the induced path of the $i$\textsuperscript{th} eigenvalue. By assumption, each $\lambda_i(t)$ moves on the real axis, and we may label eigenvalues such that $\lambda_i(0)<\lambda_j(0)$ whenever $i<j$, where being in $X$ allows for the strict inequalities.
    
    Assume a non-trivial permutation is achieved, so we may consider the smallest eigenvalue (label $i$) that gets permuted. Observe that a bigger eigenvalue (label $j$) must take its place; that is $\lambda_i(0)<\lambda_j(0)$, yet $\lambda_i(1)>\lambda_j(1)$. By the Intermediate Value Theorem, one must have $\lambda_i(t^*)=\lambda_j(t^*)$ for some $t^*\in(0,1)$. However, this implies a degeneracy which contradicts $\gamma$ being in $X$.
\end{proof}

In conclusion, we have found that encircling two or more EPs requires based oriented loops to answer the question of how the resulting permutation is composed from the permutations associated with the individual EPs.
We also described how to relate the results for different base points. In the next section we will describe a very concrete example using an experimental setup that could be used to test the results.


\section{Proposed experiment}
\label{sec:waveguide_experiment}

\subsection{Setup}

In \cite{ruter2010observation}, two waveguides are considered that are placed next to each other and coupled. They are subjected to laser pumping giving rise to a $\PT$ symmetric system in which $\PT$ phase transitions could be observed. In this paper, we investigate a three waveguide system like in \cite{schnabel2017simple}, see also the schematic picture in \cref{fig:setup}. Laser pumping induces complex refractive indices, which translates to a complex potential $V_k=N_k+iP_k$ in the $k$th channel, where $N_k=k_0n_k$ is the real refractive index part and $P_k=k_0\gamma_k/2$ the effective pumping part.

\begin{figure}[h]
    \centering
    \includegraphics{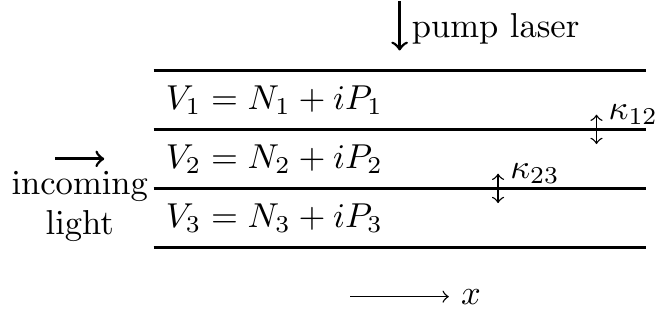}
    \caption{Schematic view of the experimental set-up (see the text).}
    \label{fig:setup}
\end{figure}

By placing channels next to each other (real) couplings $\kappa_{12},\kappa_{23}$ are induced, which depend on the coupling lengths between the channels.
The electric field amplitudes $E_k$ change along the propagation direction $x$ as (see \cite{schnabel2017ptep3} for experimental details)
\begin{equation}
    \resizebox{0.95\hsize}{!}{$
    i\frac{d}{dx}\begin{pmatrix}E_1\\E_2\\E_3\end{pmatrix}=\begin{pmatrix}V_1+iP_1&-\kappa_{12}&0\\-\kappa_{12}&V_2+iP_2&-\kappa_{23}\\0&-\kappa_{23}&V_3+iP_3\end{pmatrix}\begin{pmatrix}E_1\\E_2\\E_3\end{pmatrix}$}.
\end{equation}

Let us redefine fields and measure relative to the central channel 2.  Setting $v_k=V_k-V_2$, $p_k=P_k-P_2$, and taking equal couplings $\kappa_{12}=\kappa_{23}=\kappa$, the scaled fields $\widetilde{E}_k(x) = e^{i(V_2+iP_2)x}E_k(x)$ satisfy $i\frac{d}{dx}{\mathbf {\tilde{E}}}  = H \, {\mathbf {\tilde{E}}} $, where $ {\mathbf {\tilde{E}}} =(\widetilde{E}_1,\widetilde{E}_2,\widetilde{E}_3)^T$ and $H$ is the operator
\begin{equation}
    H=\begin{pmatrix}v_1+ip_1&-\kappa&0\\-\kappa&0&-\kappa\\0&-\kappa&v_3+ip_3\end{pmatrix}
\end{equation}
which is similar to the idealized expression found in \cite{schnabel2017simple}. The electric field components hence satisfy a Schr\"odinger type equation with a non-hermitian operator where the role of time $t$ is played by the spatial direction $x$. 

We will restrict ourselves to the subspace of operators that are of the form
\begin{equation}
    T(z,c)=\begin{pmatrix}z+2i&-\sqrt{2}&0\\-\sqrt{2}&0&-\sqrt{2}\\0&-\sqrt{2}&cz-2i\end{pmatrix} , 
\end{equation}
where $z$ is a complex and $c$ a real parameter. The cases $c=\pm1$ were investigated in \cite{demange2011signatures}, where it was shown that these are normal forms for EPs appearing in 3 dimensional systems. It was found that for $c=1$ the system has an EP3 at $z=0$, while for $c=-1$ the system has an EP2 at $z=0$. We note that $z$ is up to an offset the potential in channel 1, and $c$ is the ratio $(n_3-n_2)/(n_1-n_2)$. Observe that the whole $c$-axis ($z=0$) is mapped to the same operator. This will return in pictures in the next subsection.


\subsection{The parameter space and the discriminant set}
The parameter space of the system is the space $\mathds{C}\times\R$ with coordinates $(z,c)$. The EPs of the system are given by the parameters $(z,c)$ for which the eigenvalues of $T(z,c)$ coalesce in a branch point singularity. One can find candidates for EPs by finding higher order zeros of the characteristic polynomial $p_{z,c}(\lambda)=\det(\lambda I-T(z,c))$. The parameter space thus decomposes into a degeneracy space $\Delta$ and a non-degeneracy space $X$. $\Delta$ is then given by the discriminant set 
\[
\Delta = \{(z,c)\in \mathds{C}\times \R\;|\; \text{discrim}(p_{z,c}(\lambda),\lambda)= 0\},
\]
which forms lines in the three-dimensional parameter space. EPs can only be found on these lines, e.g.\ by finding higher order zeros of $p_{z,c}(\lambda)$ or by numerically tracking the eigenvalues along a closed loop. For the latter technique, we remark that deformation invariance of permutations allows one to check large pieces of $\Delta$ by just a single loop.

The lines in $\Delta$ contained in the plane $\mathrm{Re}(z)=0$ are shown in \cref{fig:re(z)=0plane}. Here all points in $\Delta$ are EP2s, except for points on the $c$-axis which are EP3s. This was checked by numerically following the eigenvalues along loops (the phases of the electric field were not considered). Two main features appear: a tangent intersection of two lines 
at $(0,-1)$ and a cusp at $(-4i,-1)$. The higher order zero condition implies that EP3s are confined to the $c$-axis and the cusp.

\Cref{fig:EPframe} shows what happens in the three-dimensional $(z,c)$-space, where all new lines consist of EP2s. We see that the cusp in the plane in \cref{fig:re(z)=0plane} is in fact part of a more 
complex  structure in the three-dimensional space. Here four lines move out of the plane of which two have $\mathrm{Re}(z)>0$ and the other two have $\mathrm{Re}(z)<0$. Also two additional lines of EPs appear top-right in the picture close to the central line that is already present  in \cref{fig:re(z)=0plane}.

\begin{figure}[ht]
    \centering
    \includegraphics[width=\hsize]{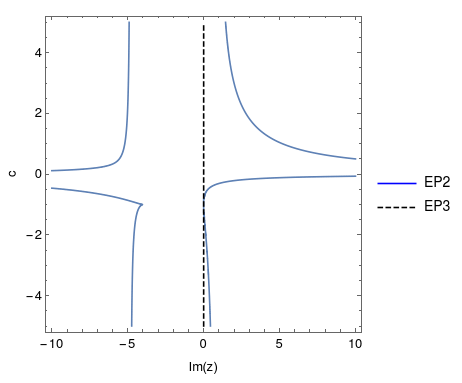}
    \caption{EP structure in the $\mathrm{Re}(z)=0$ plane.}
    \label{fig:re(z)=0plane}
\end{figure}

The tangent intersection at $(0,-1)$ and the more complex structure at $(-4i,-1)$ have another remarkable property from a $\PT$ symmetry perspective. One can define a parity operator $\P$ to swap the outer channels of the waveguide, and define a time operator $\T$ to be complex conjugation. The system is then $\PT$ symmetric at exactly three lines; the line $c=0$, the line given by $c=1$ and $z$ real, and the line given by $c=-1$ and $z$ imaginary; these were already drawn in \cref{fig:EPframe}. At the last line, $\PT$ phase transitions occur at $z=0$ and $z=-4i$, that is precisely at the tangent intersection and the point where the complex structure with several lines emerge in \cref{fig:EPframe}.

\begin{figure}
    \centering
    \includegraphics[width=\hsize]{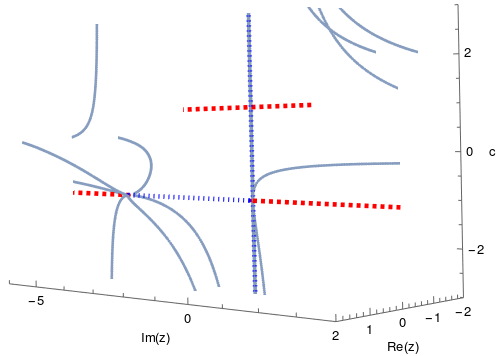}
    \caption{EP and $\PT$ structure in parameter space; the solid lines are EPs, the red blocks and blue stripes mark broken resp.\ exact $\PT$ symmetry. The picture is symmetric in $\mathrm{Re}(z)\mapsto -\mathrm{Re}(z)$.}
    \label{fig:EPframe}
\end{figure}


\subsection{The measurement}

To measure an EP, two methods stand out; one is directly tracking the eigenstates \cite{dembowski2001experimental}, which includes phase information. The other is only tracking the eigenvalues with no phase information, where the 'merging path method' is used for the identification of an EP \cite{heiss1999phases,xu2016topological}. Here one starts with a (discretized) closed path in parameter space, where one measures the eigenvalues at each point, to obtain a (discretized) locus for each eigenvalue. If no EP structure is present, the eigenvalues will individually trace out closed loops. If an EP structure is encircled, one sees the locus of one eigenvalue ending at the initial value of another one showing that eigenvalues are permuted.

Tracking eigenvalues only has clear experimental advantages; one does not need to track eigenstates adiabatically, dynamical phases are irrelevant, and slight deformation of the path yields the same permutation. The disadvantage is that the phase information may go unrecorded.

In this system, one could for fixed system parameters measure the profile of the wave in each waveguide. That is, one obtains (complex) $\tilde{E}_k(x)$ for $k=1,2,3$. Writing these in one vector $\mathbf{{\tilde{E}}}(x)$, the profiles should follow
\begin{equation}
    \mathbf{{\tilde{E}}}(x)=e^{-iHx}\mathbf{{\tilde{E}}}(x=0)
\end{equation}
in analogy to quantum mechanics. An advantage with respect to genuine quantum systems is that now gain and loss happen in space  (along the $x$-axis) and not in time.
By deducing the eigenstates $\hat E_k(x)$ (again $k=1,2,3$) theoretically, one can change to eigenstate basis. In this basis one must have 
\begin{equation}
    \hat E_k(x)=e^{-i\lambda_k x}\hat E_k(x=0).
\end{equation}
In this way the eigenvalue(s) can be obtained.


\subsection{Examples}

Let us discuss suitable paths to check the mentioned phenomena. Although other regions shown in \cref{fig:EPframe} would suffice as well, we are particularly interested in 
the region near the tangent intersection that involves both EP2s and EP3s.
The region is shown in \cref{fig:pathsneartouching}.

\begin{figure}
    \centering
    \includegraphics[width=0.9\hsize]{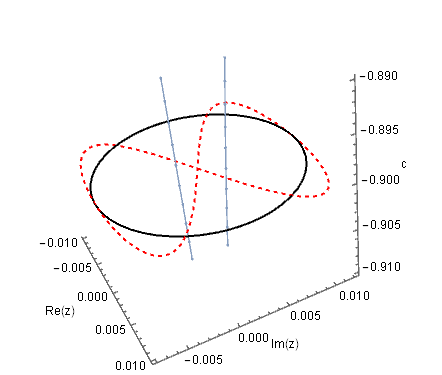}
    \caption{Paths near $(z,c)=(0,-1)$. The two blue lines are EP lines: the line $z=0$ consists of EP3s and the other line consists of EP2s. The bold black circle and dashed red figure 8 are in the plane $c=-0.9$ and can be used to experimentally verify non-abelian behaviour (see the text).}
    \label{fig:pathsneartouching}
\end{figure}

Let us first deal with the problem of concatenating loops. The relevant loops are shown in \cref{fig:composing_example}.
We deliberately take the basepoint equal in all cases, hence the slight variation on \cref{fig:theoretical_setup}. The upper EP is an EP2, the lower an EP3, taken in the plane $c=-0.9$. We  note that the situation is similar for $c$ close to this value, although the distance between the EPs varies. Hence, one can vary $c$ if it is desirable for experiment, and the discussion below will still hold.

In \cref{sfig:onlyEP2,sfig:onlyEP3} we show the employed fundamental paths $\gamma_1$ resp.\ $\gamma_2$ and find their induced permutations. As a reference, we investigate the paths $\gamma_2\gamma_1$ and $\gamma_1\gamma_2$ in resp.\ \cref{sfig:firstEP2thenEP3,sfig:firstEP3thenEP2}.
Thus, the first four pictures show the resulting permutations of eigenvalues $p_1$, $p_2$, $p_2p_1$ and $p_1p_2$, respectively. The big loop in \cref{sfig:bigcircle} is base homotopic to $\gamma_2\gamma_1$, as can be seen by pulling the left side of the loop through the area between the EPs. We observe that the permutation induced by this loop indeed equals $p_2p_1$, and does not equal $p_1p_2$. This agrees with the problem discussed in \cref{fig:basepoint_on_loop}.

Turning to the figure 8 loop in \cref{sfig:fig8}, we note that it can be deformed to $\inv{\gamma_1}\gamma_2$, so one would expect the permutation $\inv{(23)}(132)=(23)(132)=(12)$. Note that this is the same permutation as the one from $\gamma_1\gamma_2$ as $\gamma_1$ induces a transposition, yet $\gamma_1\gamma_2$ and $\inv{\gamma_1}\gamma_2$ are not homotopic. This does not contradict our claims; depending on the system, non-homotopic paths may induce the same permutation.
One can still measure orientation dependence by traversing the figure 8 in opposite direction. This loop is homotopic to $\inv{(\inv{\gamma_1}\gamma_2)}=\inv{\gamma_2}\gamma_1$, so yields the permutation $\inv{(132)}(23)=(123)(23)=(12)$. Hence this loop differs from $\gamma_2\gamma_1$ (which yields a $(13)$) only by orientation of the second part, and gives a different permutation.

\begin{figure*}
    \subfloat[Permutation $(23)$.\label{sfig:onlyEP2}
    ]{
    \adjustbox{trim={.05\width} {0\height} {0.05\width} {0\height},clip}{
    \includegraphics[width=0.49\textwidth]{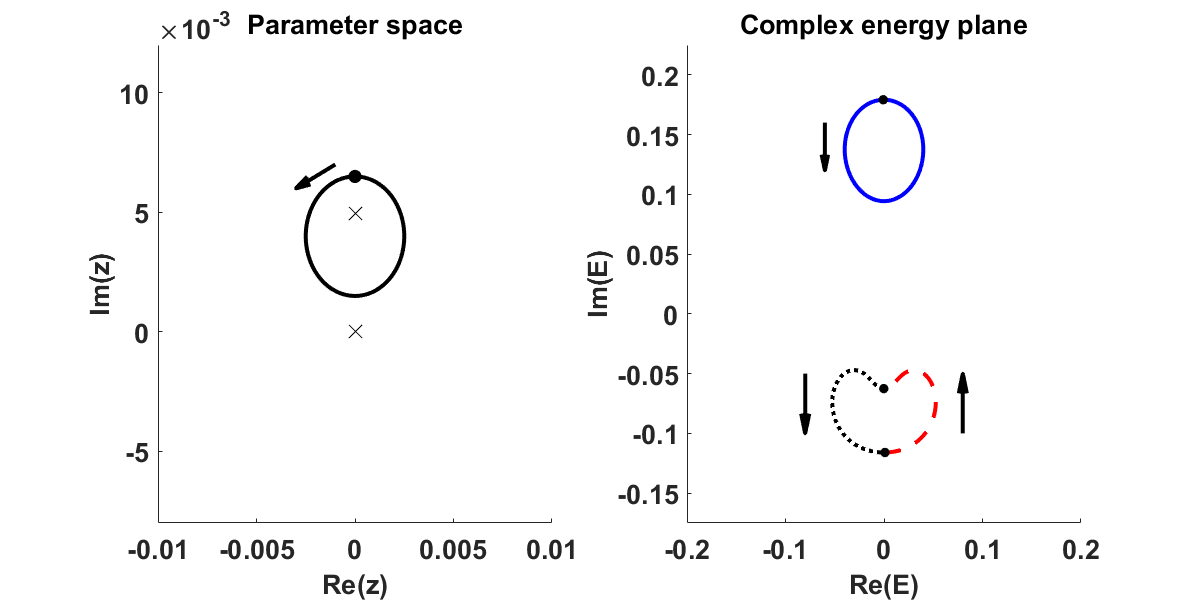}
    }
    }
    \subfloat[Permutation $(132)$.\label{sfig:onlyEP3}]{
    \adjustbox{trim={.05\width} {0\height} {0.05\width} {0\height},clip}{
    \includegraphics[width=0.49\textwidth]{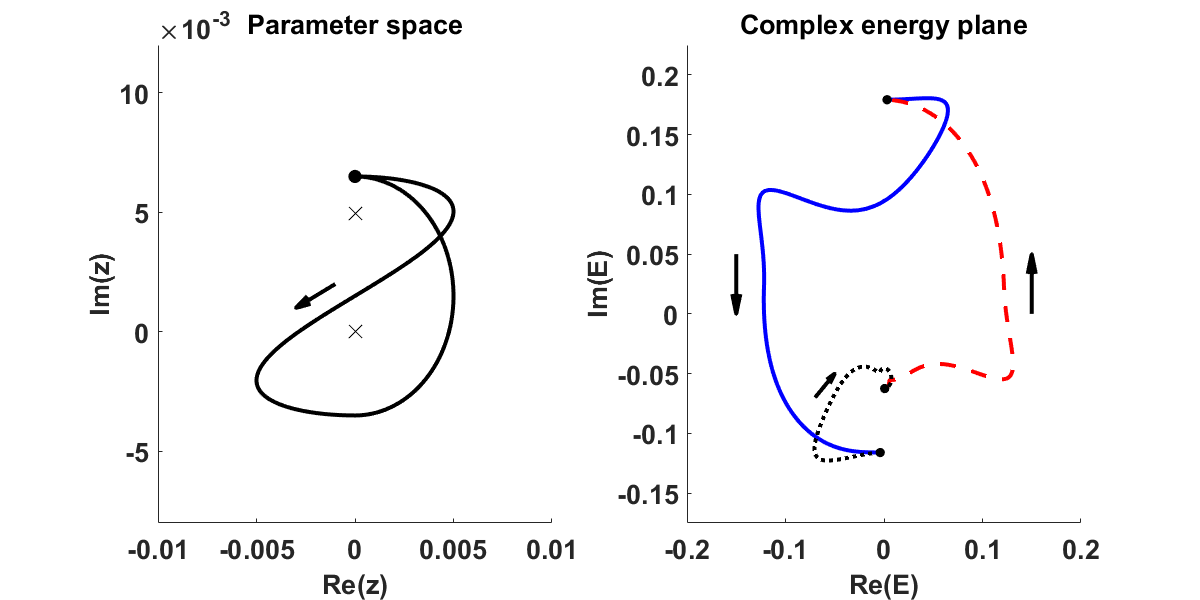}
    }
    }
    
    \subfloat[Permutation $(13)=(132)(23)$. \label{sfig:firstEP2thenEP3}]{
    \adjustbox{trim={.05\width} {0\height} {0.05\width} {0\height},clip}{\includegraphics[width=0.49\textwidth]{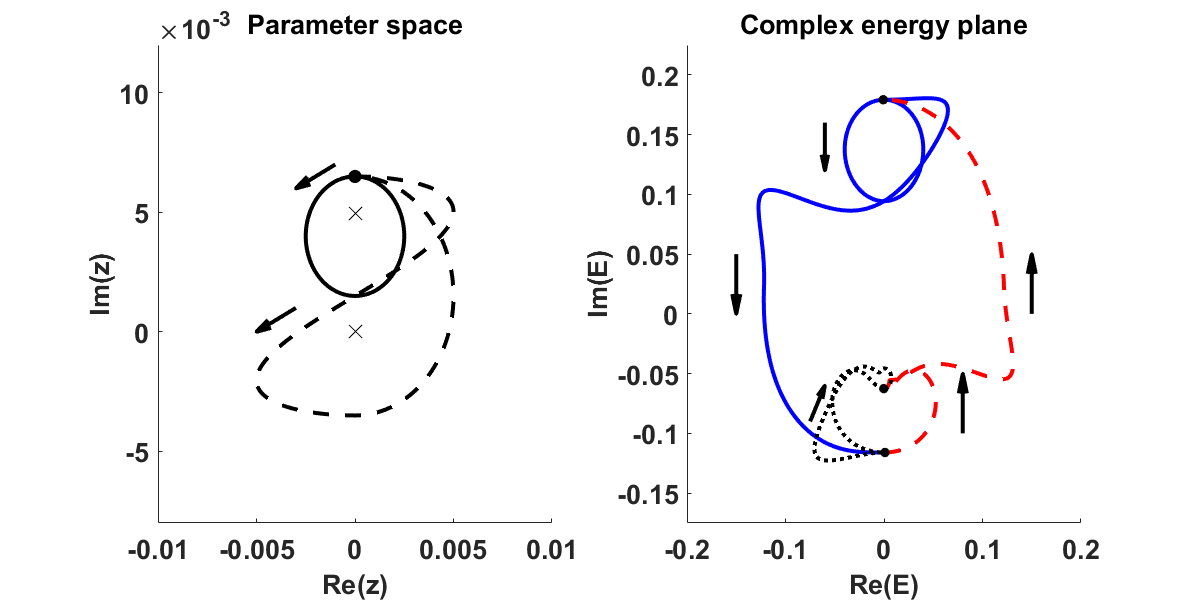}
    }
    }
    \subfloat[Permutation $(12)=(23)(132)$. \label{sfig:firstEP3thenEP2}
    ]{
    \adjustbox{trim={.05\width} {0\height} {0.05\width} {0\height},clip}{\includegraphics[width=0.49\textwidth]{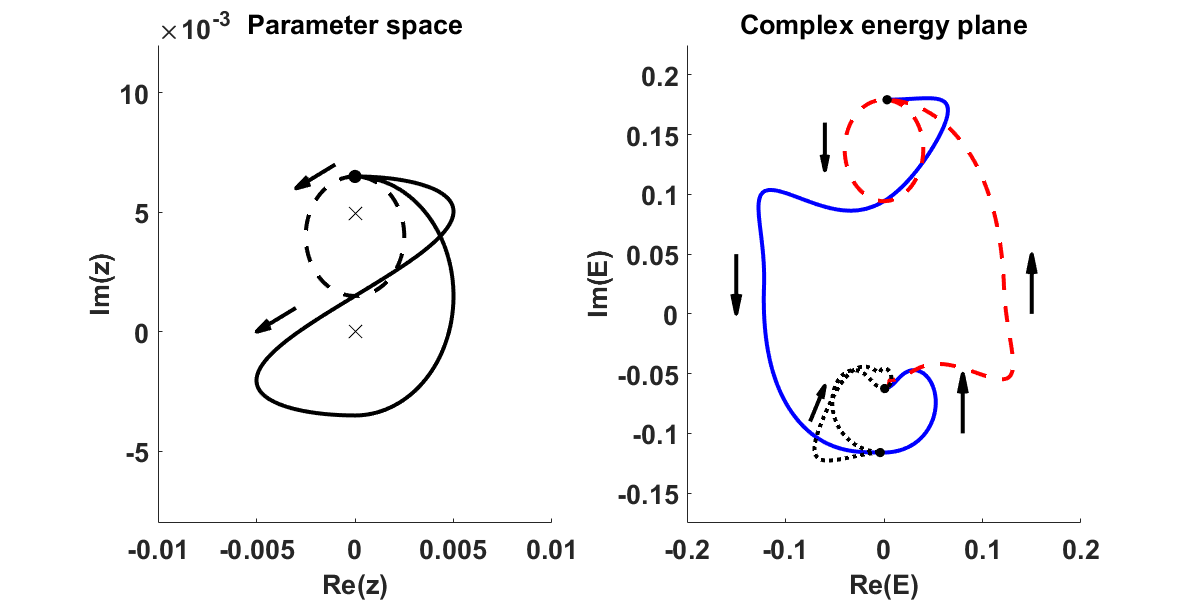}
    }
    }
    
    \subfloat[Permutation $(13)$. \label{sfig:bigcircle}
    ]{
    \adjustbox{trim={.05\width} {0\height} {0.05\width} {0\height},clip}{\includegraphics[width=0.49\textwidth]{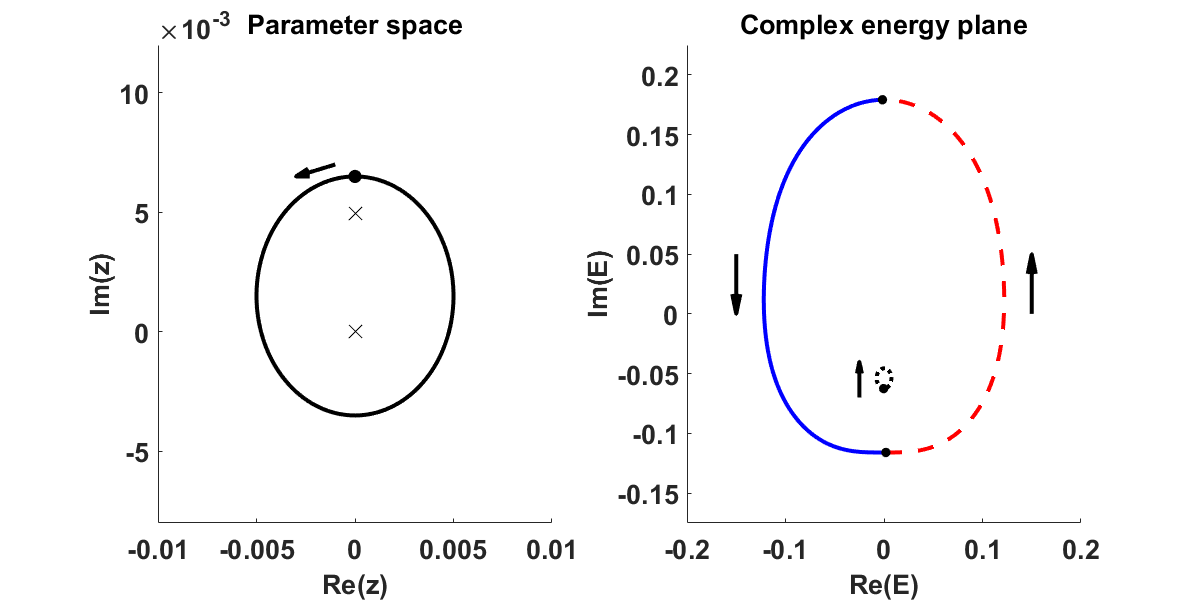}
    }
    }
    \subfloat[Permutation $(12)$. \label{sfig:fig8}]{
    \adjustbox{trim={.05\width} {0\height} {0.05\width} {0\height},clip}{\includegraphics[width=0.49\textwidth]{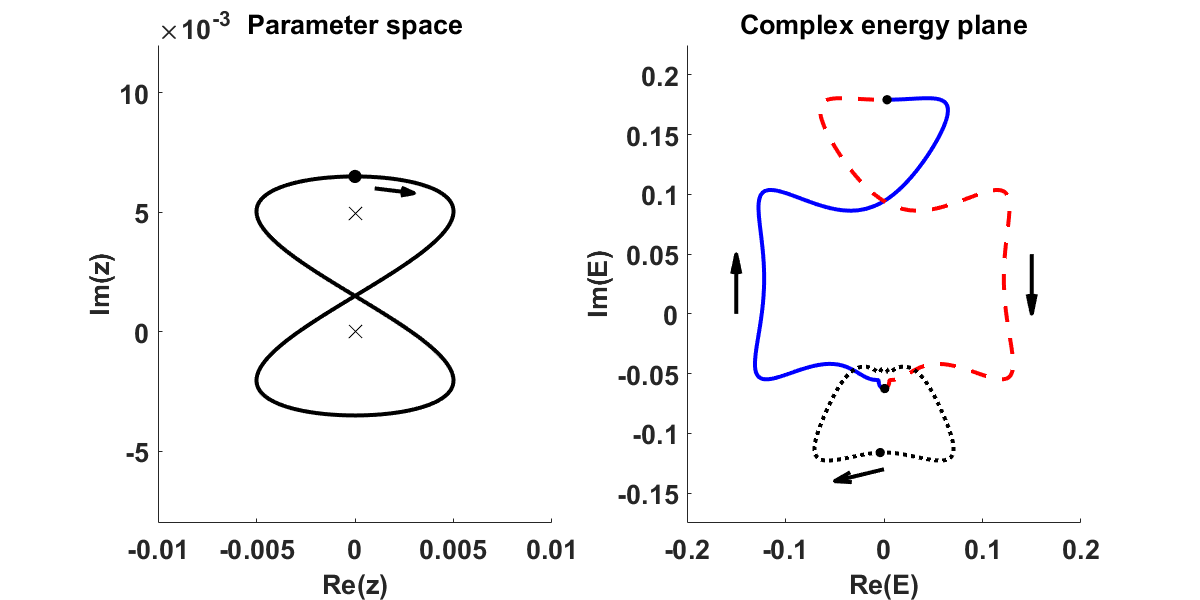}
    }
    }

    \caption{Loops in the plane plane $c=-0.9$ of the parameter space (left panels). All loops have the same basepoint, marked by a circle. The solid part of the loop is always traversed first, so before the dashed part of the loop. The orientations of the loops are indicated by arrows. The EPs are marked by crosses, where the upper one is an EP2 and the lower one is an EP3. The complex energy planes (right panels) show the resulting paths of the three eigenvalues, each drawn with its own color and style. Labelling eigenvalues top to bottom, we can read off the induced permutation given in the individual captions. One observes that the $\Lambda$-group of the system is isomorphic to $S_3$.}
    \label{fig:composing_example}
\end{figure*}


\subsection{The degree of an exceptional point in systems with more than two parameters}

Taking a closer look at the tangent intersection, one may ask the question what its degree should be. As reported in \cite{demange2011signatures}, the tangent intersection may behave as an EP2. This means that traversing a circle in the $c=-1$ plane which encircles this EP (and only this EP) yields the standard EP2 signature of swapping 2 eigenstates, as shown in \cref{fig:touchingpointasEP2}, and as expected resembles the result of \cref{sfig:bigcircle} (using the obvious relabeling). However, the four lines arrive in a topological cross, and one may take a circle that goes through the other two quadrants. In this case, one can take a plane given by $\mathrm{Im}(z)=\epsilon i$ with $\epsilon>0$ small, and take a large circle. Interestingly, this yields the standard EP3 signature, as seen in \cref{fig:touchingpointasEP3}.

\begin{figure*}
    \subfloat[The tangent intersection as EP2. \label{fig:touchingpointasEP2}
    ]{
    \includegraphics[width=0.48\textwidth]{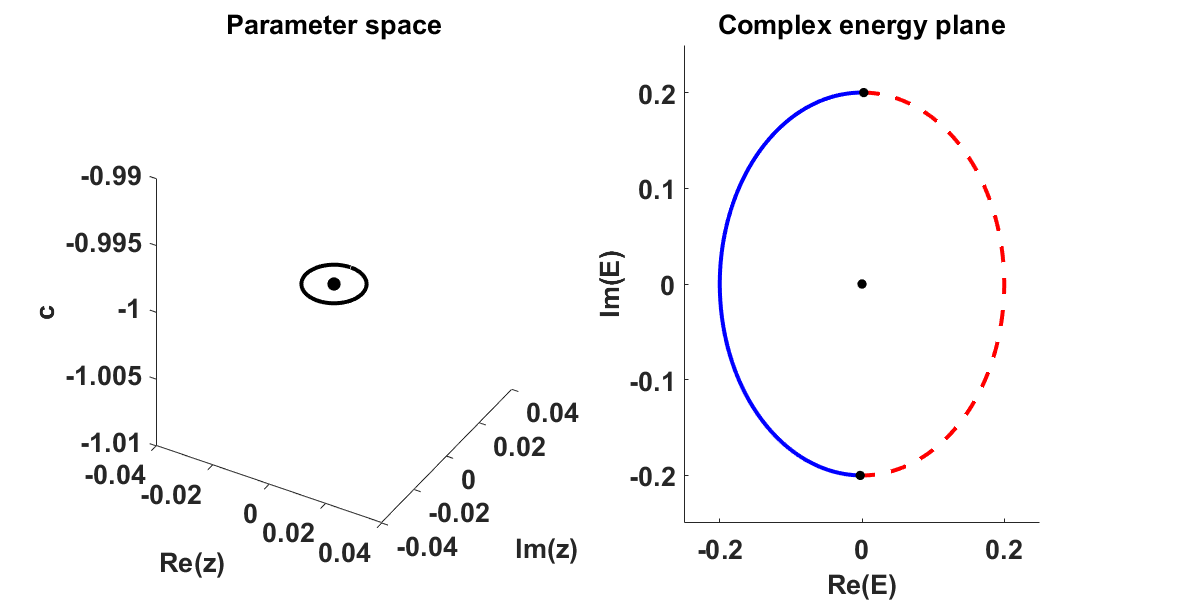}
    }
    \subfloat[The tangent intersection as EP3. \label{fig:touchingpointasEP3}
    ]{
    \includegraphics[width=0.48\textwidth]{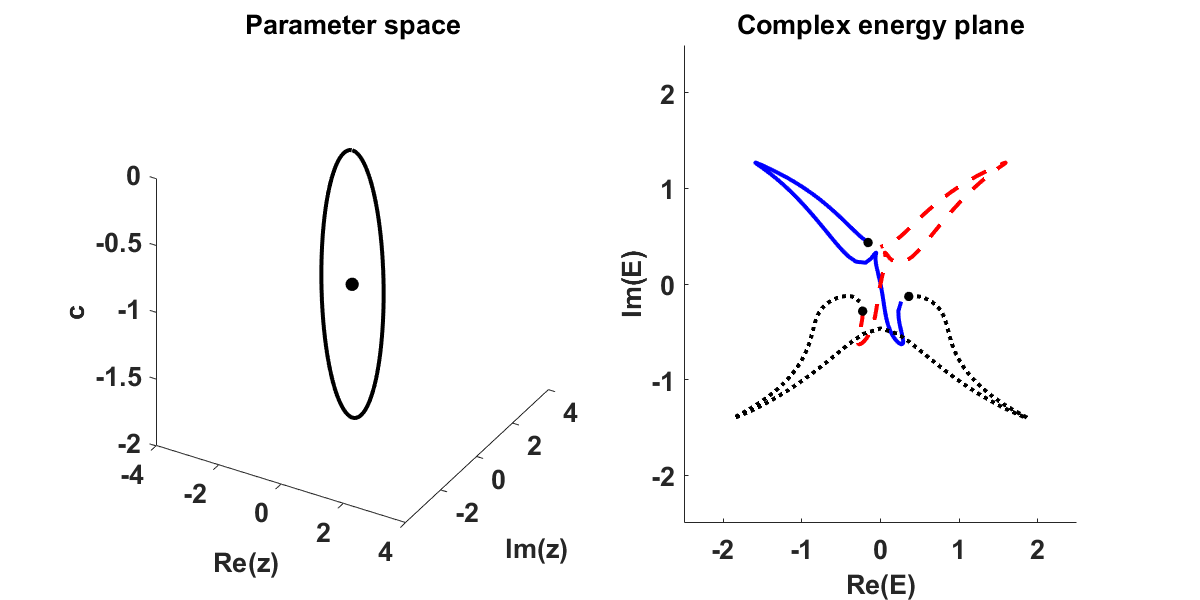}
    }

    \caption{The tangent intersection as EPs of different orders. \label{fig:ill_defined_EP_degree}}
\end{figure*}

One can now do a similar construction with the point at $(-4i,-1)$ and conclude that its degree depends on the plane. In a general parameter space of dimension $d>2$, any point where at least 3 EP lines meet has a variable degree (note that the degree of an EP is unambiguous on the lines).

Again fundamental groups provide an explanation. In case of a planar parameter space with an isolated EP, the fundamental group is $\Z$ and one has a map $\Z\to \Lambda(x_0)$. This has kernel $N\Z$, and $N$ is the degree of the EP. Now, imagine 2 distinct EP structures/lines, as we want an intersection point necessarily for $d\geq3$. The fundamental group is then the free product $\Z *\Z$, i.e.\ generated by 2 fundamental paths $\gamma_1$ and $\gamma_2$. Hence each $\gamma_i$ induces a map $\Z\to \Lambda(x_0)$, with kernel $N_1\Z$ resp. $N_2\Z$. In case $N_1\ne N_2$ clearly an issue arises, but even if $N_1=N_2$ we see that degree should be associated to a fundamental path, or equivalently some surface. At the intersection point there is simply no canonical choice.


\section{Summary}
\label{sec:summary}

We showed how one can compose the effects obtained from encircling multiple EPs, which in fact works for an arbitrary degeneracy structure. The problem of finding the correct calculation can be solved by using the theory of fundamental groups, which requires based oriented loops. A relevant result here is that permutations associated to the loops are of topological as opposed to geometric nature, hence deformation can be used for convenience in both theory and experiment.

Applications of these insights were explored in a waveguide system, of which we investigated the parameter space and identified a region where all tests could be performed. The presence of both EP2s and EP3s allows one to demonstrate the non-abelian nature of systems with multiple EPs by experimentally tracking the eigenvalues.


%

\end{document}